\documentclass[final,1p,times,tikz,border=5pt]{article}
\usepackage{tikz-cd}
\usepackage{tikz}
\usepackage{extpfeil}

\usepackage{amssymb}
\usepackage{amsthm}
\usepackage{pgfplots}
\usepackage{latexsym}
\usepackage{amsmath}
\usepackage{color}
\usepackage{graphicx}
\usepackage{indentfirst}
\usepackage{mathrsfs}
\usepackage{lipsum}
\usepackage{comment}
\usepackage{graphicx}
\usepackage{subcaption} 
\pgfplotsset{compat=1.18}
\usepackage{physics}
\allowdisplaybreaks

\newtheorem{theorem}{\color{black}\indent \textbf{Theorem}}[section]

\newtheorem{proposition}{\color{black}\indent Proposition}[section]
\newtheorem{definition}{\color{black}\indent Definition}[section]
\newtheorem{remark}{\color{black}\indent Remark}[section]

\newcommand{\ii}{\iota}

\title{Integration on $q$-Cosymplectic Manifolds}

\author{%
  \parbox{\linewidth}{\centering
    Melvin Leok$^{a}$, Cristina Sardón$^{b}$, Xuefeng Zhao$^{c}$\\[1ex]
    $^{a}$Department of Mathematics, University of California, San Diego, 
    9500 Gilman Drive, Dept.\ 0112, La Jolla, CA 92093-0112, USA\\
    $^{b}$Department of Applied Mathematics, Universidad Politécnica de Madrid, 
    Av.\ Juan de Herrera 6, 28040, Madrid, Spain\\
    $^{c}$College of Mathematics, Jilin University, Changchun, 130012, P.R.\ China\\[1ex]
    \texttt{mleok@ucsd.edu}, \texttt{mariacristina.sardon@upm.es}, \texttt{zhaoxuef@jlu.edu.cn}
  }%
}
		
\begin{document}
\maketitle
	\begin{abstract}
			{This paper presents a unified framework for studying dynamics and integration on \( q \)-cosymplectic manifolds. After outlining the geometric foundations of \( q \)-cosymplectic structures, we derive new results concerning integrable systems and the characterization of Liouville coordinates, and further investigate the Lie integrability of $q$-evolution systems in this setting. We then develop a Hamilton--Jacobi theory tailored to multi-time Hamiltonian systems, both from an intrinsic geometric perspective and via symplectification techniques. To illustrate the applicability of the framework, we construct a \( q \)-cosymplectic Hamiltonian model for an extended FitzHugh--Nagumo system, providing a biologically relevant example involving three distinct temporal scales.}
		\end{abstract}
		

\section{Introduction}

The interplay between geometry and dynamics has been a central theme in mathematics and physics since the development of Hamiltonian mechanics. Classical symplectic geometry provides the natural framework for time-independent Hamiltonian systems, while cosymplectic geometry {\cite{Berceanu,Chinea,Fino,Li,Rovenski,Shafiee}} extends this setting to incorporate dynamics with a distinguished time variable. However, many modern applications require the analysis of systems evolving on multiple timescales or depending on several independent ``time-like'' variables. To accommodate such structures, the concept of a \emph{$q$-cosymplectic manifold} has been introduced as a higher-codimensional generalization of cosymplectic geometry \cite{LeokSardonZhao2025,deLeonBajo2025}.  

A $q$-cosymplectic structure endows a smooth manifold $M^{2n+q}$ with a closed two-form $\Omega$, together with a family of $q$ closed one-forms $(\lambda_1,\ldots,\lambda_q)$, and a splitting of the tangent bundle into a Reeb distribution and a horizontal distribution. This framework naturally produces commuting Reeb vector fields $(R_1,\ldots,R_q)$, which generate flows along independent temporal directions. As in the classical case, the horizontal distribution inherits a symplectic structure from $\Omega$, thereby allowing Hamiltonian dynamics to be formulated in the multi-time setting.  

Generalizations of classical integrability theory, such as Liouville’s theorem and the Arnold–Liouville torus action~\cite{Liouville, Arnold2}, have been extended to the setting of $q$-cosymplectic systems. In particular, recent results establish the existence of Liouville coordinates and invariant tori under the presence of commuting strict $q$-cosymplectic vector fields, thereby providing a natural notion of integrability within this framework.{In addition, the concept of Lie integrability \cite{Azuaje,Carinena,Kozlov} has been generalized to \( q \)-evolution vector field on $q$-cosymplectic manifold.}
 Furthermore, the Hamilton–Jacobi theory {\cite{Esen,Esen2,Ferraro,Leon,Ohsawa}}, which plays a crucial role in both the geometric integration of dynamical systems and in control theory, admits a natural extension to the $q$-cosymplectic context. This extension can be formulated either intrinsically on the $q$-cosymplectic manifold, or via symplectification into a canonical cotangent bundle \cite{adjoint2}.  

Beyond their intrinsic geometric interest, $q$-cosymplectic structures arise naturally in models of physical and biological systems with multiple temporal scales. For instance, FitzHugh–Nagumo type neuron models exhibit fast, intermediate, and slow dynamics, which can be encoded geometrically within a 3-cosymplectic structure \cite{fitz1,fitz2,fitz3,fitz4}. This allows the use of Hamiltonian and Hamilton–Jacobi methods to analyze the rich multi-scale behavior of such systems, including reductions, invariant manifolds, and slow-fast bifurcations.  

The purpose of this paper is to develop a unified framework for \emph{dynamics and integration on $q$-cosymplectic manifolds}. After reviewing the geometric foundations of $q$-cosymplectic structures, we establish results on integrable systems and Liouville coordinates {as well as Lie integrability} in this setting. We then formulate the Hamilton–Jacobi theory for multi-time Hamiltonian systems, both intrinsically and through symplectification. Finally, we illustrate the theory by constructing a $q$-cosymplectic Hamiltonian description of an extended FitzHugh–Nagumo system, which provides a biologically meaningful application involving three distinct timescales.

    \section{Geometric Fundamentals}

Let \( n, q \) be positive integers, and let \( M \) be a smooth manifold of dimension \( 2n + q \). A \( q \)-cosymplectic structure on \( M \) is a collection consisting of a closed 2-form \( \Omega \) and a tuple of \( q \) (pointwise) linearly independent, non-vanishing, closed 1-forms \( \vec{\lambda} = (\lambda_1, \dots, \lambda_q) \), together with a splitting
$TM = \mathcal{R} \oplus \xi$ of the tangent bundle, satisfying the following conditions:

\begin{itemize}
    \item[(i)] \( \xi := \bigcap_{i=1}^q \ker \lambda_i; \)
    
    \item[(ii)] There exists a unique collection of linearly independent vector fields \( R_1, \dots, R_q \), tangent to \( \mathcal{R} \), satisfying the relations
    \[
    \lambda_i(R_j) = \delta_{ij}, \quad \text{for all } i,j = 1, \dots, q;
    \]
    
    \item[(iii)] \( \mathcal{R} = \mathrm{span}\{ R_1, \dots, R_q \}; \)
    
    \item[(iv)] \( \ker \Omega = \mathcal{R} \), and the restriction \( \Omega|_\xi \) is non-degenerate.
\end{itemize}

The vector fields \( R_i \), \( i = 1, \dots, q \), are called the \emph{Reeb vector fields}, and the forms \( \lambda_i \), \( i = 1, \dots, q \), are called the \emph{cosymplectic forms}.  The linear independence of the \( \lambda_i \) implies that \( \xi \) has constant rank \( 2n \). Therefore, by condition (iv), \( (\xi, \Omega) \) defines a symplectic vector bundle over \( M \).
It is easy to prove that if the vector fields \( R_1, \dots, R_q \) satisfy conditions \emph{(ii)} and \emph{(iii)} {\cite{Leok}}, then
\[
[R_i, R_j] = 0, \quad \text{for all } i,j = 1, \dots, q.
\]
The most straightforward example of $q$-cosymplectic manifold is the product
$M = T_1 \times \cdots \times T_q \times P,$ where \( T_i \), \( i = 1, \dots, q \), are one-dimensional manifolds, and \( (P, \Omega) \) is a symplectic manifold. Defining the projections \( \pi_{T_i} : M \rightarrow T_i \) and \( \pi_P : M \rightarrow P \), the symplectic form \( \Omega \) on \( P \) induces a closed 2-form on \( M \) given by $\Omega_P := \pi_P^* \Omega.$ Similarly, a non-vanishing 1-form \( \eta_i \) on \( T_i \) induces a closed 1-form \( \eta_{iT} := \pi_{T_i}^* \eta_i \) on \( M \). Then, $M = T_1 \times \cdots \times T_q \times P, \; \Omega_P, \; \vec{T} = (\eta_{1T}, \dots, \eta_{qT})$ defines a \( q \)-cosymplectic manifold.

A manifold endowed with such a structure is called a \emph{\( q \)-cosymplectic manifold} and is denoted by \( (M, \Omega, \vec{\lambda}) \), or simply by \( M \) when the context is clear. We refer to the collection \( \{ \lambda_i \} \) as an \emph{adapted coframe} for the \( q \)-cosymplectic structure, and the \( q \)-form $\lambda := \lambda_1 \wedge \cdots \wedge \lambda_q \neq 0$ is called the \emph{characteristic form}. The bundles \( \mathcal{R} \) and \( \xi \) are called the \emph{Reeb distribution} and the \emph{\( q \)-cosymplectic distribution}, respectively. 

Now, let  \( X \) be a vector field on a $q$-cosymplectic manifold  \( (M, \Omega, \vec{\lambda}) \), with associated flow \( \psi_t \), for all \( t \in \mathbb{R} \).
We say that  \( X \) is a \emph{\( q \)-cosymplectic vector field} (or \emph{infinitesimal automorphism of \( \xi \)}) if \( T\psi_t(\xi) = \xi \) for all \( t \), and we call \( X \) a \emph{strict \( q \)-cosymplectic vector field} or \emph{infinitesimal automorphism of \( \lambda_1 \wedge \cdots \wedge \lambda_q \)} if $\psi_t^*(\lambda_1 \wedge \cdots \wedge \lambda_q) = \lambda_1 \wedge \cdots \wedge \lambda_q, \forall t.$
    
The most simple example of a  \( q \)-cosymplectic structure on \( \mathbb{R}^{2n+q} \) with coordinates 
\[
(x_1, y_1, \dots, x_n, y_n, z_1, \dots, z_q),
\]
defined by
\[
\lambda_i := dz_i, \quad \Omega = \sum_{j=1}^n x_j \, dy_j,
\]
and
\[
\mathcal{R} = \mathrm{span} \left\{ \frac{\partial}{\partial z_1}, \dots, \frac{\partial}{\partial z_q} \right\}.
\]
The \( q \)-cosymplectic distribution is given by
\[
\xi = \mathrm{span} \left\{ \frac{\partial}{\partial x_1}, \frac{\partial}{\partial y_1}, \dots, \frac{\partial}{\partial x_n}, \frac{\partial}{\partial y_n} \right\}.
\]
We observe that the vector fields \( \frac{\partial}{\partial z_i} \), \( i = 1, \dots, q \), are Reeb vector fields. Moreover, they are strict \( q \)-cosymplectic vector fields, since
\[
\mathcal{L}_{\frac{\partial}{\partial z_i}} \lambda_j = 0, \quad \forall i,j = 1, \dots, q.
\]

A Bogoyavlenskii type theorem on a $q$-cosymplectic manifold \( M \) guarantees that given a submersion \( F = (F_1, \dots, F_k) : M \rightarrow \mathbb{R}^k \) with compact and connected fibers, for some \( 1 \leq k < \dim M \), a collection of \( n = \dim M - k \) vector fields \( Y_1, \dots, Y_n \) that are everywhere linearly independent, pairwise commuting, and tangent to the fibers of \( F \), i.e.,
    \[
    [Y_i, Y_j] = 0, \quad \mathcal{L}_{Y_i} F_r = 0, \quad \forall i,j = 1, \dots, n, \; r = 1, \dots, k.
    \]

Then, the map \( F : M \rightarrow F(M) \subset \mathbb{R}^k \) defines a \( \mathbb{T}^n \)-bundle and any vector field \( X \) on \( M \) satisfying
    \[
    \mathcal{L}_X F_r = 0 \quad \text{and} \quad [X, Y_i] = 0, \quad \forall i = 1, \dots, n, \; r = 1, \dots, k,
    \]
    is conjugate to a constant vector field on \( \mathbb{T}^n \) via each bundle chart of \( F : M \rightarrow F(M) \).

So, the tuple \( (Y_1, \dots, Y_n, F_1, \dots, F_k) \) is called an \emph{integrable system of type \( (n, k) \)} on \( M \). Therefore, a system on a manifold \( M \) is called \emph{integrable} if there exists an integrable system \( (Y_1, \dots, Y_n, F_1, \dots, F_k) \) of some type \( (n, k) \) on \( M \) with \( Y_1 = Y \).

Note that the vector fields \( Y_1, \dots, Y_n \) are tangent to the fibers of \( F \). We say that the system \( (Y_1, \dots, Y_n, F_1, \dots, F_k) \) is \emph{regular} at a fiber \( L \) of \( F \) if
\[
Y_1 \wedge \cdots \wedge Y_n \neq 0 \quad \text{and} \quad dF_1 \wedge \cdots \wedge dF_k \neq 0 \quad \text{everywhere on } L.
\]
We say that the system is \emph{proper} if the map \( (F_1, \dots, F_k): M \rightarrow \mathbb{R}^k \) is a proper topological map (i.e., each level set is compact), and the system is regular on almost every fiber.

Given a smooth manifold \( M \), we define \emph{natural vector bundles} over \( M \) to be those vector bundles that can be constructed from the tangent bundle \( TM \), the cotangent bundle \( T^*M \), and the trivial line bundle \( \mathbb{R} \times M \) using operations such as direct sums, tensor products, and their iterates.

There exists a Liouville type theorem \cite{Liouville} for $q$-cosymplectic manifolds. Consider the integrable system \( (Y_1, \dots, Y_n, F_1, \dots, F_k) \) is regular at a compact level set \( L \) of \( F \). Then, in a tubular neighborhood \( \mathcal{U}(L) \), there exists, up to automorphisms of \( \mathbb{T}^n \), a unique free torus action
\[
\rho : \mathbb{T}^n \times \mathcal{U}(L) \rightarrow \mathcal{U}(L),
\]
which preserves the system (that is, the action preserves each \( Y_i \) and each \( F_j \)), and whose orbits are regular level sets of the system. In particular, \( L \) is diffeomorphic to \( \mathbb{T}^n \), and
\[
\mathcal{U}(L) \cong \mathbb{T}^n \times B^k,
\]
with periodic coordinates \( \theta_1 \, (\mathrm{mod}\, 1), \dots, \theta_n \, (\mathrm{mod}\, 1) \) on the torus \( \mathbb{T}^n \), and coordinates \( (z_1, \dots, z_k) \) on a \( k \)-dimensional ball \( B^k \), such that \( F_1, \dots, F_k \) depend only on the variables \( z_1, \dots, z_k \), and the vector fields \( Y_i \) are of the form
\[
Y_i = \sum_{j=1}^n a_{ij}(z_1, \dots, z_k) \, \frac{\partial}{\partial \theta_j}.
\]

A system of coordinates
\[
(\theta_1 \; (\mathrm{mod}\, 1), \dots, \theta_n \; (\mathrm{mod}\, 1), z_1, \dots, z_k)
\]
on \( \mathcal{U}(L) \cong \mathbb{T}^n \times B^k \), as given by the above theorem, is called a \emph{Liouville system of coordinates}. 

Due to the previous theorem, each \( n \)-dimensional compact level set \( L \) of an integrable system of type \( (n, k) \), on which the system is regular, is called a \emph{Liouville torus}, and the torus \( \mathbb{T}^n \)-action in a tubular neighborhood \( \mathcal{U}(L) \) of \( L \) that preserves the system is called the \emph{Liouville torus action}. Notice that this action is uniquely determined by the system, up to an automorphism of \( \mathbb{T}^n \).

Furthermore, we can recall the fundamental conservation property  \cite{Zung} for integrable $q$-cosymplectic systems. Consider the Liouville torus \( L \) of an integrable system \( (Y_1, \dots, Y_n, F_1, \dots, F_k) \) on a manifold \( M \), and let \( \mathcal{G} \in \Gamma(\otimes^h TM \otimes^k T^*M) \) be a tensor field on \( M \) that is preserved by all the vector fields of the system:
\[
\mathcal{L}_{Y_i} \mathcal{G} = 0, \quad \forall i = 1, \dots, n.
\]
Then the Liouville torus \( \mathbb{T}^n \)-action on a tubular neighborhood \( \mathcal{U}(L) \subset M \) also preserves \( \mathcal{G} \).

\begin{definition}
An integrable system \( (Y_1, \dots, Y_n, F_1, \dots, F_k) \) on a \( 2n + q \)-dimensional \( q \)-cosymplectic manifold \( (M, \vec{\lambda}, \mathcal{R} \oplus \xi) \) is called a \emph{\( q \)-cosymplectic integrable system} if the vector fields \( Y_1, \dots, Y_n \) are strict \( q \)-cosymplectic vector fields.
\end{definition}
To understand the local geometric structure of 
q-cosymplectic integrable systems, it is natural to examine the behavior of the system in a neighborhood of a Liouville torus. The following theorem shows that, in such a neighborhood, one can introduce a system of Liouville coordinates in which the structural forms $\lambda_i$ are invariant under the $\mathbb T^n$-action and admit a canonical expression in terms of angle and action variables.
\begin{theorem}
If \( (Y_1, \dots, Y_n, F_1, \dots, F_k) \) is a \( q \)-cosymplectic integrable system on a \( 2n + q \)-dimensional \( q \)-cosymplectic manifold \( (M, \Omega, \vec{\lambda}, \mathcal{R} \oplus \xi) \), then in a neighborhood \( \mathcal{U}(N) \cong \mathbb{T}^n \times B^k \) of any Liouville torus \( N \subset M \), the forms \( \lambda_i \) are \( \mathbb{T}^n \)-invariant under the Liouville coordinate system \( (\theta_i \; (\mathrm{mod}\; 1), z_j) \), and are expressed as
\[
\lambda_i = \sum_{j=1}^n a_{ij}(z) \, d\theta_j + \sum_{j=1}^k b_{ij}(z) \, dz_j, \quad i = 1, \dots, q.
\]

\end{theorem}

\begin{proof}
The characteristic form \( \lambda_1 \wedge \cdots \wedge \lambda_q \) is \( \mathbb{T}^n \)-invariant under the Liouville coordinate system \( (\theta_i \; (\mathrm{mod}\; 1), z_j) \). Therefore, the characteristic form must be of the form
\[
\lambda_1 \wedge \cdots \wedge \lambda_q = \sum_{\substack{h + l = q \\ 0 \leq h \leq n, \; 0 \leq l \leq k}} f_{h,l}(z) \, d\theta_{i_1} \wedge \cdots \wedge d\theta_{i_h} \wedge dz_{j_1} \wedge \cdots \wedge dz_{j_l},
\]
for some smooth functions \( f_{h,l}(z) \) depending only on \( z \)-variables. This implies that each \( \lambda_i \) must be of the form
\[
\lambda_i = \sum_{j=1}^n a_{ij}(z) \, d\theta_j + \sum_{j=1}^k b_{ij}(z) \, dz_j,
\]
for some functions \( a_{ij}(z), b_{ij}(z) \), which confirms that each \( \lambda_i \) is \( \mathbb{T}^n \)-invariant.
%
\end{proof}

\section*{Integration on $q$-cosymplectic manifolds}

Let $(M, \Omega, \lambda_1, \dots, \lambda_q)$ be a {q-cosymplectic manifold}. The manifold $M$ is the {phase space} of dimension $2n+q$, and $N$ is the {configuration space} of dimension $n+q$. The projection $\pi: M \to N$ is the canonical projection onto $(x^j, t^i)$ coordinates. For example, in local coordinates $(x^j, p_j, t^i)$, $M$ is the total space of momenta $p_j$, while $N$ contains only the positions and time parameters. Given a Hamiltonian function $H \in C^\infty(M)$, the Hamiltonian vector field $X_H$ is defined by:
\[
i_{X_H}\Omega = dH - \sum_{i=1}^q R_i(H)\lambda_i, \quad i_{X_H}\lambda_i = 0.
\]

The {q-evolution vector field} is:
\[
E_H = \sum_{i=1}^q R_i + X_H.
\]

On a \( q \)-cosymplectic \( (M, \Omega, \vec{\lambda}) \) we can define the bracket
\[
\{ \cdot, \cdot \} : C^\infty(M) \times C^\infty(M) \to C^\infty(M), \quad (f, g) \mapsto \{f, g\} := \Omega(X_f, X_g),
\]
it is easy to prove this bracket is a Poisson bracket and 
further we can verify  for any Hamiltonian vector field $X_f$
\[
i_{X_f} \lambda_i = 0, \quad i = 1, \dots, q.
\]
Indeed, by the definition of the map \( b \), we have
\[
b(X_f) = i_{X_f} \Omega + \sum_{i=1}^q \lambda_i(X_f) \lambda_i = df - \sum_{i=1}^q (R_i f) \lambda_i.
\]
Therefore, for any \( R_j \), it follows that
\begin{align*}
\lambda_j(X_f) 
&= i_{R_j} i_{X_f} \Omega + i_{R_j} \left( \sum_{i=1}^q \lambda_i(X_f) \lambda_i \right) \\
&= i_{R_j} df - i_{R_j} \left( \sum_{i=1}^q (R_i f) \lambda_i \right) 
= R_j(f) - R_j(f) = 0.
\end{align*}
Hence, we can get 
\begin{align*}
    \{f,g\}=\Omega(X_f,X_g)=i_{X_g}i_{X_f}\Omega=i_{X_g}\left(df-\sum_{i=1}^qR_i(f)\lambda_i\right)=X_g(f).
\end{align*}

Now, let \(\mathfrak{g}\) be a solvable Lie algebra, i.e., for \(\mathfrak{g}\)  there exists a chain of Lie subalgebras, let us say $L_1,...,L_n$, such that
$\{0\} = L_0 \subset L_1 \subset L_2 \subset \cdots \subset L_{n-1} \subset L_n = \mathfrak{g}, \quad \text{with } n = \dim(\mathfrak{g}),$ and that each \(L_i\) is an ideal in \(L_{i+1}\) with codimension one, for \(i = 0,1,\dots,n-1\). We briefly recall here the Lie integrability theorem on a general manifold  \cite{Arnold2}. For \(X_1, \dots, X_n\) linearly independent smooth vector fields on $n$ dimensional manifold \(M\), generating a solvable Lie algebra \(\mathfrak{g}\), if one of the vector fields (denoted \(X_1\)) defines a dynamical system and that all vector fields are symmetries of \(X_1\), i.e., 
	\[
	[X_1, X_i] = 0, \quad \text{for } i = 2, \dots, n.
	\]
	Then the dynamical system defined by $X_1$ can be solved by quadratures. Generalizing this theorem to $q$-cosymplectic structures, we have the following.

\begin{theorem} \label{keji}
Consider the flow of the \( q \)-evolution vector field \( E_H = \sum_{i=1}^q R_i + X_H \) on a \( 2n+q \)-dimensional \( q \)-cosymplectic manifold \( (M, \Omega, \vec{\lambda}) \). Assume the following:

\begin{enumerate}
\item The vector field \( E_H \) admits \( m \) independent first integrals \( f_1, \dots, f_m \) and some constants $c_{ij}^k$ such that 
\[
\{f_i, f_j\} =\sum_{k=1}^r c_{ij}^kf_k, \quad  \{f_i,f_l\}=0, \;\text{for all}\;i,j=1,...,r < m,\;l=r+1,...,m
\]
and the condition \( 2n + q - 1 = m + r \) holds.

	\item The functions \(f_1, \ldots, f_r\) generate a solvable Lie algebra under the Poisson bracket;

\item The map \( F := (f_{r+1}, \dots, f_m) : M \to \mathbb{R}^{m - r} \) is a submersion whose fibers
\[
M_{\mathbf{c}} = \{x \in M \mid f_{r+1}(x) = c_1, \dots, f_m(x) = c_{m - r}\}, \quad \mathbf{c} \in \mathbb{R}^{m - r},
\]
are compact and connected.

	\item The structure constants satisfy
		\[
		\sum_{k=1}^rc_{ij}^k c_k = 0, \quad \text{for all } i, j = 1, \ldots, r;
		\]
\end{enumerate}

Then each fiber \( M_{\mathbf{c}} \) is diffeomorphic to a torus \( \mathbb{T}^{m - r} \), and the vector field \( E_H \) restricted to \(M_{\mathbf{c}}\) can be obtained by quadratures.
\end{theorem}
Before showing the proof of this theorem, we present a brief review of the concept of functional independence and its consequences. Consider a smooth manifold $N,$ we say that the functions $F_1,...,F_k$ with $k\leq$ dim $N$ are functionally independent at the point $p\in N$ if the differential maps $dF_1|_p,dF_2|_p,...,dF_k|_p$ are linearly independent, which means that $p$ is a regular point of the differentiable function $F:N\rightarrow \mathbb R^k$ defined by $F=(F_1,...,F_k).$ Consider the set $N_{f}=\{x\in N:F_i(x)=c_i,c_i\in\mathbb R\},$ and suppose that $F_1,...,F_k$ are functionally independent, from the regular level set theorem \cite{Lee}, we have that $N_{f}$ is a smooth submanifold of $N$ of dimension dim $N-k.$
	\begin{proof}
	The functional independence of \(f_1, \ldots, f_m\) on \(M_{\mathbf{c}}\) implies that \(M_{\mathbf{c}}\) is a smooth submanifold of dimension \(2n + q - (m - r)\).

For any Hamiltonian functions \(f\) and \(g\), we have the identity
\[
X_{\{f, g\}} = [X_g, X_f].
\]
Applying this to our setting, we obtain
\[
X_{\{f_i, f_j\}} = [X_{f_j}, X_{f_i}], \quad \text{for all } i, j = 1, \ldots, r.
\]
Moreover, we know that for any function \(f \in C^\infty(M)\),
\begin{align}
[R_j, X_f] = X_{R_j(f)}, \quad j = 1, \dots, q.
\end{align}
Thus, we have
\begin{align*}
[E_H, X_{f}] &= \left[X_H + \sum_{j=1}^q R_j, X_{f} \right] = -X_{\{H, f\}} + \sum_{j=1}^q X_{R_j(f)} \\
&= X_{X_H(f)} + X_{\sum_{j=1}^q R_j(f)} = X_{E_H(f)} = 0.
\end{align*}
So, the vector fields \(E_H, X_{f_1}, \ldots, X_{f_r}\) generate a solvable Lie algebra under the Lie bracket of vector fields.

Next, observe that each \(X_{f_i}\), for \(i = 1, \ldots, r\), is tangent to \(M_{\mathbf{c}}\). Indeed, for \(i = 1, \ldots, r\) and \(j = r+1, \ldots, m\),
\begin{align}
(X_{f_i} f_j)\big|_{M_{\mathbf{c}}} = \{f_j, f_i\}\big|_{M_{\mathbf{c}}} = -\sum_{k=1}^r c_{ij}^k f_k\big|_{M_{\mathbf{c}}} = -\sum_{k=1}^r c_{ij}^k c_k = 0.
\end{align}

Therefore, the flow of \(X_{f_1}\) remains within \(M_{\mathbf{c}}\), and the dynamics restricted to \(M_{\mathbf{c}}\) are governed by the differential equation
\[
\dot{x} = X_{f_1}(x),
\]
where \((x^1, \ldots, x^k)\) are local coordinates on \(M_f\).
Since \(X_{f_1}\) belongs to a solvable Lie algebra generated by \(X_{f_1}, \ldots, X_{f_k}\), the system is solvable by quadratures. That is, its solutions can be obtained through successive integrations.

	\end{proof}
\section*{Hamilton-Jacobi Problem}
\noindent
Consider a {section} $\gamma: N \to M$ such that $\pi \circ \gamma = \text{id}_N$. We seek a vector field $X$ on $N$ satisfying:
\[
T\gamma \circ X = X_H \circ \gamma.
\]

This ensures that the integral curves of $X$ lift via $\gamma$ to integral curves of $X_H$. This means that the dynamics can be solved in a lower dimensional manifold and the complete dynamic can be reconstructed via the section $\gamma.$ We can try the following trick. We pull back the defining equation of $X_H$ via $\gamma$:
\[
\gamma^*(i_{X_H}\Omega) = \gamma^*\left( dH - \sum_{i=1}^q R_i(H)\lambda_i \right).
\]

Since $T\gamma \circ X = X_H \circ \gamma$, we have:
\[
\gamma^*(i_{X_H}\Omega) = i_X(\gamma^*\Omega).
\]

If $\gamma$ is {isotropic} (which is the condition we need to impose in this case), then $\gamma^*\Omega = 0$, so:
\[
0 = d(H \circ \gamma) - \sum_{i=1}^q (R_i(H) \circ \gamma)\, \gamma^*\lambda_i.
\]

That is,
\[
d(H \circ \gamma) = \sum_{i=1}^q (R_i(H) \circ \gamma)\, \gamma^*\lambda_i.
\]
This is the geometric Hamilton--Jacobi equation for multi-time systems with a compatible $q$-cosymplectic structure.

The necessary and sufficient condition for the projectability of the dynamics in this case is that the image of $\gamma$ is an isotropic submanifold of $M$.
We will say that $\text{Im}(\gamma) \subset M$ is isotropic if:
\[
\Omega|_{\text{Im}(\gamma)} = 0.
\]
That is, for all $p \in \text{Im}(\gamma)$ and all $v, w \in T_p \text{Im}(\gamma)$, $\Omega_p(v,w) = 0$.

In our context, $\gamma^*\Omega = 0$ means that $\text{Im}(\gamma)$ is isotropic. The 1-form $dH - \sum_{i=1}^q R_i(H)\,\lambda_i$ {annihilates all the tangent vectors in the image of} $\gamma$.
In other words, the dynamics of $X_H$ are tangent to the image of $\gamma$ (modulo the Reeb directions). This condition ensures that the integral curves of a vector field $X$ on $N$, when lifted via $\gamma$, become integral curves of $X_H$ on $M$.

\begin{remark}
     If \(\Omega\) is degenerate, this is referred to as {presymplectic orthogonality}. See that it is degenerate in the direct sum of $\xi \oplus \mathcal{R}$, since the Reebs are in the kernel of $\Omega$.

Now, let \(E = T(\operatorname{Im}(\gamma)) \subset \xi\), and assume \(\gamma^* \Omega = 0\), so \(E\) is isotropic. since \(\Omega|_\xi\) is symplectic, we can compute the orthogonal inside \(\xi\) as in symplectic geometry:
\[
E^\perp_\xi := \{ v \in \xi \mid \Omega(v, w) = 0 \ \text{for all } w \in E \}.
\]

Then, in the full tangent bundle \(TM = \mathcal{R} \oplus \xi\), the full \(\Omega\)-orthogonal complement is:
\[
E^\perp = E^\perp_\xi \oplus \mathcal{R}.
\]

This is because \(\Omega(v, w) = 0\) for all \(w \in E\) if:
\begin{itemize}
  \item \(v \in \xi\) and \(\Omega(v, w) = 0\), i.e., \(v \in E^\perp_\xi\), or
  \item \(v \in \mathcal{R}\), in which case \(\Omega(v, \cdot) = 0\) anyway.
\end{itemize}

So, combining these:
\[
E^\perp = E^\perp_\xi \oplus \mathcal{R}.
\]
\noindent
Now, since \(E \subset \xi\) is isotropic, we have \(E \subset E^\perp_\xi\). So:
\[
E^\perp_\xi = E + E^{\Omega, \mathrm{skew}},
\]
and therefore:
\[
E^\perp = E \oplus \mathcal{R} + \text{possibly more in } \xi.
\]

If the dimension of \(E\) is maximal isotropic (i.e., half the rank of \(\Omega|_\xi\)), then:
\[
E^\perp_\xi = E,
\quad \text{so} \quad
E^\perp = E \oplus \mathcal{R}.
\]
\noindent
So the lifted dynamics lie in the orthogonal complement, which turns out to be \(T(\operatorname{Im}(\gamma)) \oplus \mathcal{R}\) in the HJ theory.
\end{remark}

Visually, imagine that $\gamma$ draws a ``surface'' in $M$, parametrized by $N$. The vector field $X_H$ is compatible with moving along $\gamma(N)$, in the sense that its horizontal part is completely determined by the base dynamics.

The condition:
\[
\gamma^*\left( dH - \sum_{i=1}^q R_i(H)\lambda_i \right) = 0
\]
is the {mathematical expression of this compatibility}.
Let us prove this assertion using the commutative diagram defined by the projectability condition.

\subsubsection*{Hamilton-Jacobi formulation with commutative diagrams}

The projection of dynamics is represented by the following commutative diagram:

\[
\begin{tikzcd}
M \arrow{r}{X_H} & TM \\
N \arrow{u}{\gamma} \arrow{r}{X} & TN \arrow{u}{T\gamma}
\end{tikzcd}
\]

This expresses:
\[
T\gamma \circ X = X_H \circ \gamma.
\]

Let \( M \) be the phase space of dimension \( 2n + q \) with coordinates \( (x^j, p_j, t^i) \), and \( N \) the configuration space of dimension \( n + q \) with coordinates \( (x^j, t^i) \). The projection \( \pi: M \to N \) is given by:
\[
\pi(x^j, p_j, t^i) = (x^j, t^i).
\]

A section \( \gamma: N \to M \) satisfying \( \pi \circ \gamma = \text{id}_N \) is locally of the form:
\[
\gamma(x^j, t^i) = \left( x^j, \gamma_j(x, t), t^i \right),
\]
where \( \gamma_j \in C^\infty(N) \) are the momenta determined by the section. It is a geometric condition: the image of \( \gamma \) must be an \emph{isotropic submanifold} of \( (M, \Omega) \). In this case, since \( \Omega = dx^j \wedge dp_j \), the pullback becomes $\gamma^* \Omega = dx^j \wedge d\gamma_j$. So the condition \( \gamma^* \Omega = 0 \) means $dx^j \wedge d\gamma_j = 0,$ which is equivalent to requiring that the 1-form \( \gamma_j\, dx^j \) is closed $d(\gamma_j\, dx^j) = 0$.
(In classical terms, this corresponds to choosing a solution of the Hamilton--Jacobi equation coming from a \emph{generating function} \( S(x) \), such that $\gamma_j = {\partial S}/{\partial x^j}$. In this case, the 1-form \( \gamma_j\, dx^j \) is automatically exact (and hence closed), so the isotropy condition is satisfied.)

\begin{proposition}
Let \((M, \Omega)\) be a symplectic manifold with \(\Omega = dx^j \wedge dp_j\), and let \(N\) be a base manifold with coordinates \((x^j, t^i)\). Consider a section \(\gamma : N \to M\) of the projection \(\pi : M \to N\), locally given by
\[
\gamma(x^j, t^i) = \big(x^j, \gamma_j(x, t), t^i\big),
\]
where the functions \(\gamma_j \in C^\infty(N)\) specify the momenta. Then the condition that the image of \(\gamma\) is an isotropic submanifold of \((M, \Omega)\), i.e.,
\[
\gamma^* \Omega = 0,
\]
is equivalent to requiring:
\begin{align*}
\frac{\partial \gamma_j}{\partial x^k} &= \frac{\partial \gamma_k}{\partial x^j} \quad \text{for all } j,k, \\
\frac{\partial \gamma_j}{\partial t^l} &= 0 \quad \text{for all } j,l.
\end{align*}
Moreover, this implies that there exists a local function \(S(x)\) such that
\[
\gamma_j(x) = \frac{\partial S}{\partial x^j}.
\]
\end{proposition}

\begin{proof}
We compute the pullback of \(\Omega = dx^j \wedge dp_j\) under the section \(\gamma\). Since \(\gamma_j = p_j \circ \gamma\), we have:
\[
\gamma^* \Omega = dx^j \wedge d\gamma_j.
\]
Define the 1-form \(\theta = \gamma_j dx^j\) on \(N\). Then,
\[
d\theta = d(\gamma_j dx^j) = d\gamma_j \wedge dx^j.
\]
In local coordinates, we compute:
\[
d(\gamma_j dx^j) = \sum_{j,k} \frac{\partial \gamma_j}{\partial x^k} dx^k \wedge dx^j + \sum_{j,l} \frac{\partial \gamma_j}{\partial t^l} dt^l \wedge dx^j.
\]
Using antisymmetry of the wedge product, we group the terms:
\[
d(\gamma_j dx^j) = \sum_{j<k} \left( \frac{\partial \gamma_j}{\partial x^k} - \frac{\partial \gamma_k}{\partial x^j} \right) dx^k \wedge dx^j + \sum_{j,l} \frac{\partial \gamma_j}{\partial t^l} dt^l \wedge dx^j.
\]
Thus, the condition \(\gamma^* \Omega = 0\) (i.e., \(d\theta = 0\)) implies:
\[
\frac{\partial \gamma_j}{\partial x^k} = \frac{\partial \gamma_k}{\partial x^j}, \quad \frac{\partial \gamma_j}{\partial t^l} = 0.
\]
These are precisely the conditions for the 1-form \(\theta = \gamma_j dx^j\) to be closed and independent of \(t\). By the Poincaré lemma, locally there exists a function \(S(x)\) such that:
\[
\gamma_j(x) = \frac{\partial S}{\partial x^j}.
\]
\end{proof}

In this setting, the canonical forms are:
\[
\Omega = dx^j \wedge dp_j, \qquad \lambda_i = dt^i,
\]
and the Reeb vector fields are \( R_i = \frac{\partial}{\partial t^i} \). Given a Hamiltonian function \( H \in C^\infty(M) \), the Hamiltonian vector field \( X_H \) is defined by:
\[
i_{X_H} \Omega = dH - \sum_{i=1}^q R_i(H) \lambda_i, \qquad i_{X_H} \lambda_i = 0.
\]

Let us now write down the expression of \( X_H \) in local coordinates. Assume:
\[
X_H = a^j \frac{\partial}{\partial x^j} + b_j \frac{\partial}{\partial p_j} + c^i \frac{\partial}{\partial t^i}.
\]

Then:
\[
i_{X_H} \Omega = a^j dp_j - b_j dx^j, \quad i_{X_H} \lambda_i = c^i = 0.
\]
\noindent
Hence, from:
\[
dH - \sum_i R_i(H) \lambda_i = \frac{\partial H}{\partial x^j} dx^j + \frac{\partial H}{\partial p_j} dp_j,
\]
we find:
\[
a^j = \frac{\partial H}{\partial p_j}, \qquad b_j = -\frac{\partial H}{\partial x^j}, \qquad c^i = 0.
\]

So the Hamiltonian vector field is:
\[
X_H = \frac{\partial H}{\partial p_j} \frac{\partial}{\partial x^j} - \frac{\partial H}{\partial x^j} \frac{\partial}{\partial p_j}.
\]

We seek a vector field \( X \) on \( N \) such that:
\[
T\gamma \circ X = X_H \circ \gamma.
\]

Let:
\[
X = \frac{\partial H}{\partial p_j} \circ\gamma\frac{\partial}{\partial x^j}.
\]

Then:
\[
T\gamma\left( \frac{\partial}{\partial x^j} \right) = \frac{\partial}{\partial x^j} + \frac{\partial \gamma_k}{\partial x^j} \frac{\partial}{\partial p_k},
\]
so:
\[
T\gamma(X) = \frac{\partial H}{\partial p_j} \circ\gamma \left( \frac{\partial}{\partial x^j} + \frac{\partial \gamma_k}{\partial x^j} \frac{\partial}{\partial p_k} \right).
\]

On the other hand,
\[
X_H \circ \gamma = \left( \frac{\partial H}{\partial p_j} \circ \gamma \right) \frac{\partial}{\partial x^j} - \left( \frac{\partial H}{\partial x^j} \circ \gamma \right) \frac{\partial}{\partial p_j}.
\]

Matching components gives:
\[
\sum_j \left( \frac{\partial H}{\partial p_j} \circ \gamma \right) \frac{\partial \gamma_k}{\partial x^j} = -\left( \frac{\partial H}{\partial x^k} \circ \gamma \right),
\]
which is equivalent to:

\[
d(H \circ \gamma) = \sum_{j=1}^{q} \left( \frac{\partial H}{\partial t_j} \circ \gamma \right) dt_j.
\]

In fact, using the condition
\[
\gamma(x) = (\gamma_1(x), \ldots, \gamma_n(x)) \quad \text{with} \quad {\frac{\partial \gamma_j}{\partial x^k} = \frac{\partial \gamma_k}{\partial x^j},}
\]
we have
\begin{align*}
\left( \frac{\partial H}{\partial p_j} \circ \gamma \right) {\frac{\partial \gamma_k}{\partial x^j}}
&= { \left( \frac{\partial H}{\partial p_j} \circ \gamma \right) \frac{\partial \gamma_j}{\partial x^k}},
\end{align*}
which means that
\[
\sum_j \left( \frac{\partial H}{\partial p_j} \circ \gamma \right) \frac{\partial \gamma^j}{\partial x^k} + \left( \frac{\partial H}{\partial x^k} \circ \gamma \right) = 0.
\]
So,
\begin{align*}
d(H \circ \gamma) 
&= dH(x, \gamma(x), t) = dH(x, \gamma_1(x), \ldots, \gamma_n(x), t) \\
&= \sum_{j=1}^{n} \left( \frac{\partial H}{\partial x^j} \circ \gamma \right) dx^j + \sum_{i=1}^{n} \left( \frac{\partial H}{\partial p_i} \circ \gamma \right) \sum_{j=1}^{q} \frac{\partial \gamma^i}{\partial x^j} dx^j + \sum_{j=1}^{q} \left( \frac{\partial H}{\partial t_j} \circ \gamma \right) dt_j.
\end{align*}
{Hence}, we can get that
\[
d(H \circ \gamma) - \sum_{j=1}^{q} \left( \frac{\partial H}{\partial t_j} \circ \gamma \right) dt_j = {\sum_{j=1}^{q} \left(\sum_{i=1}^q \left( \frac{\partial H}{\partial p_i} \circ \gamma \right) \frac{\partial \gamma^i}{\partial x^j} + \left( \frac{\partial H}{\partial x^j} \circ \gamma \right) \right)} = 0.
\]

Since \( \Omega = dx^j \wedge dp_j \), we have:
\[
\gamma^* \Omega = dx^j \wedge d\gamma_j,
\]
and thus, if \( \gamma \) is isotropic, i.e., \( \gamma^* \Omega = 0 \), the differential 1-form \( d(H \circ \gamma) - \sum_i (R_i(H) \circ \gamma) \gamma^* \lambda_i \) must vanish. {Since} the pullback of the defining equation of \( X_H \) gives:
\[
\gamma^*(i_{X_H} \Omega) = d(H \circ \gamma) - \sum_{i=1}^q (R_i(H) \circ \gamma) \gamma^* \lambda_i,
\]
{if} \( \gamma^* \Omega = 0 \), then:
\begin{equation}\label{HJequationqco}
0 = \gamma^*(i_{X_H} \Omega) = d(H \circ \gamma) - \sum_{i=1}^q (R_i(H) \circ \gamma) \gamma^* \lambda_i.
\end{equation}

We now show how the geometric Hamilton--Jacobi equation on a $q$-cosymplectic manifold can be derived via symplectification. 

\subsubsection*{HJ on $q$-cosymplectic via symplectification}
Let $(M, \Omega, \lambda_1, \ldots, \lambda_q)$ be a $q$-cosymplectic manifold, with 
\[
M = T^*Q \times \mathbb{R}^q, \quad \Omega = dx^j \wedge dp_j, \quad \lambda_i = dt^i,
\]
where $Q$ is the configuration space with coordinates $x^j$, and $t^i$ represent $q$ independent time variables. The Reeb vector fields are given by $R_i = \partial / \partial t^i$.

To symplectify $M$, we consider the cotangent bundle of the extended configuration space:
\[
\widetilde{M} := T^*(Q \times \mathbb{R}^q) = T^*Q \times T^*\mathbb{R}^q,
\]
with coordinates $(x^j, p_j, t^i, \tau_i)$ and the canonical symplectic form
\[
\widetilde{\Omega} = dx^j \wedge dp_j + d\tau_i \wedge dt^i.
\]
We view $M$ as embedded in $\widetilde{M}$ via the inclusion $\iota : M \hookrightarrow \widetilde{M}$ given by $\tau_i = 0$. Let $\pi : \widetilde{M} \to M$ denote the projection map defined by
\[
\pi(x^j, p_j, t^i, \tau_i) = (x^j, p_j, t^i).
\]

Given a Hamiltonian function $H(x, p, t) \in C^\infty(M)$, the associated Hamiltonian vector field $X_H$ on $M$ is defined by:
\[
i_{X_H} \Omega = dH - \sum_{i=1}^q R_i(H) dt^i.
\]
This relation also holds on $\widetilde{M}$, but it can be rewritten in the symplectic framework as follows:
\begin{align}\label{EQi}
i_{X_H} \Omega + \sum_{i=1}^q R_i(H) dt^i = dH 
\quad \Longleftrightarrow \quad
i_{X_H} \Omega + i_{\sum_{i=1}^q R_i(H) \frac{\partial}{\partial \tau_i}} \left( \sum_{j=1}^q d\tau_j \wedge dt^j \right) = dH.
\end{align}
Define the vector field $\widetilde{X}_H$ on $\widetilde{M}$ by
\[
\widetilde{X}_H := X_H + \sum_{i=1}^q R_i(H) \frac{\partial}{\partial \tau_i}.
\]
Then the equation \eqref{EQi} becomes:
\[
i_{\widetilde{X}_H} \widetilde{\Omega} = dH,
\]
which shows that $\widetilde{X}_H$ is the Hamiltonian vector field associated to $H$ on $(\widetilde{M}, \widetilde{\Omega})$. Moreover, since
\[
T\pi \circ \widetilde{X}_H = X_H,
\]
we immediately obtain the following proposition:

\begin{proposition}
Each Hamiltonian vector field on a $q$-cosymplectic manifold can be lifted to a Hamiltonian vector field on the symplectification of that manifold.
\end{proposition}

We now establish a connection between the Hamilton--Jacobi theory formulated on a $q$-cosymplectic manifold and its counterpart on the symplectified space.

Let $S(x^j, t^i)$ be a function on $Q \times \mathbb{R}^q$. Define a section $\widetilde{\gamma} : Q \times \mathbb{R}^q \to \widetilde{M}$ by:
\[
\widetilde{\gamma}(x, t) = \left( x^j, \frac{\partial S}{\partial x^j}, t^i, \frac{\partial S}{\partial t^i} \right).
\]
This section is Lagrangian because
\begin{align}\label{GG}
\widetilde{\gamma}^* \widetilde{\Omega} = d\left( \frac{\partial S}{\partial x^j} \right) \wedge dx^j + d\left( \frac{\partial S}{\partial t^i} \right) \wedge dt^i = d^2 S = 0
\end{align}
We also require a vector field $X \in \mathfrak{X}(Q \times \mathbb{R}^q)$ such that
\begin{align}\label{GG2}
T\widetilde{\gamma} \circ X = \widetilde{X}_H \circ \widetilde{\gamma}.
\end{align}
Under conditions \eqref{GG} and \eqref{GG2}, we obtain the multi-time Hamilton--Jacobi equation in the symplectic setting:
\[
H \circ \widetilde{\gamma} = \text{constant},
\]
which is equivalent to
\[
\widetilde{\gamma}^* dH = 0.
\]

We refer to this as the classical multi-time Hamilton--Jacobi equation.

Now, define the projected section:
\[
\gamma(x, t) := \pi \circ \widetilde{\gamma}(x, t) = \left( x^j, \frac{\partial S}{\partial x^j}, t^i \right) \in M.
\]
Using the section $\gamma$ induced from $\widetilde{\gamma}$, we aim to derive the Hamilton--Jacobi equation for $X_H$ on the $q$-cosymplectic manifold.

In our previous discussion, we identified two conditions for $\gamma$ to yield a solution to the Hamilton--Jacobi equation on a $q$-cosymplectic manifold:{
\begin{enumerate}
    \item $\gamma^* \Omega = 0$;
    \item There exists a vector field $Y \in \mathfrak{X}(Q \times \mathbb{R}^q)$ such that $T\gamma \circ Y = X_H \circ \gamma$.
\end{enumerate}}

Condition (ii) follows immediately from condition \eqref{GG2}. In fact, applying $T\pi$ to both sides of equation \eqref{GG2} gives:
\[
T\pi \circ (T\widetilde{\gamma} \circ X) = T\pi \circ (\widetilde{X}_H \circ \widetilde{\gamma})
\quad \Longleftrightarrow \quad
T(\pi \circ \widetilde{\gamma}) \circ X = (T\pi \circ \widetilde{X}_H) \circ (\pi \circ \widetilde{\gamma}).
\]
Since $T\pi \circ \widetilde{X}_H = X_H$ and $\gamma = \pi \circ \widetilde{\gamma}$, we obtain:
\[
T\gamma \circ X = X_H \circ \gamma,
\]
so we may simply take $Y := X$.

We now analyze when condition (i), i.e., $\gamma^* \Omega = 0$, holds.

From equation \eqref{GG}, we have:
\[
0 = \widetilde{\gamma}^* \widetilde{\Omega} 
= \widetilde{\gamma}^* \left( \Omega + \sum_{i=1}^q d\tau_i \wedge dt^i \right)
= \widetilde{\gamma}^* \Omega + \widetilde{\gamma}^* \left( \sum_{i=1}^q d\tau_i \wedge dt^i \right).
\]
Moreover, since $\Omega = \pi^* \Omega$, we find:
\[
0 = \gamma^* \Omega = (\pi \circ \widetilde{\gamma})^* \Omega = \widetilde{\gamma}^* (\pi^* \Omega) = \widetilde{\gamma}^* \Omega
\]
if and only if
\begin{align}\label{GG3}
0 = \widetilde{\gamma}^* \left( \sum_{i=1}^q d\tau_i \wedge dt^i \right)
= \sum_{i=1}^q d\left( \frac{\partial S}{\partial t^i} \right) \wedge dt^i. 
\end{align}
Equation \eqref{GG3} is equivalent to:
\[
\frac{\partial^2 S}{\partial t^i \partial x^j} = 0, \quad \text{for all } i = 1, \ldots, q,\; j = 1, \ldots, n.
\]

\begin{remark}\label{REM}
The condition
\[
\frac{\partial^2 S}{\partial t^i \partial x^j} = 0, \quad i = 1, \ldots, q,\; j = 1, \ldots, n,
\]
implies that each function $\, {\frac{\partial S}{\partial x^j}} \,$ depends only on the variables $x^1, \ldots, x^n$. This corresponds precisely to the fact that, in the case of a $q$-cosymplectic manifold, the section $\gamma$ depends only on the $x$-variables.
\end{remark}

\section{Application}

Consider the following system inspired by extended FitzHugh–Nagumo dynamics \cite{fitz1}

\[
\begin{aligned}
\varepsilon \dot{x} &= x - \frac{x^3}{3} - y, \\
\delta \dot{y} &= x + a - z, \\
\dot{z} &= b x - c z,
\end{aligned}
\]
\noindent
where \( 0 < \varepsilon \ll \delta \ll 1 \), and \( a, b, c \in \mathbb{R} \) are parameters. This system exhibits {three distinct time scales}:
\begin{itemize}
    \item The variable \( x \) evolves on the \emph{fast time scale} \( t_f = t / \varepsilon \),
    \item The variable \( y \) evolves on the \emph{intermediate time scale} \( t_i = t / \delta \),
    \item The variable \( z \) evolves on the \emph{slow time scale} \( t_s=t \).
\end{itemize}
\noindent
This classifies the variable speeds as $x \approx \mathcal{O}(1/\varepsilon)$ being the fast variable, $y$ is the intermediate fast variable $y \approx \mathcal{O}(1/\delta)$ and $z$ the slow scale with time scale $z \approx \mathcal{O}(1)$. 
These systems arise in models of bursting neurons, where \( x \) represents membrane potential (fast),\( y \) represents recovery/adaptation (intermediate) and  \( z \) models ultra-slow regulatory mechanisms (e.g., ion concentration).

 The neuron's membrane voltage ($x$) has an initial depolarization (the voltage increases from its resting value), followed by repolarization as the neuron relaxes back.  
    Because the parameters are below the spiking threshold, the trajectory does not reach the regime of sustained oscillations.  
    Instead, the membrane potential makes only a small excursion and then stabilizes, corresponding to a quiescent neuron at rest.
The depolarizing ionic current ($y$) counteracts the fast voltage change.  When the $x$ depolarization rises, $y$ also increases, acting as a braking force that helps repolarize the membrane.  This reflects the action of fast ionic mechanisms (such as sodium inactivation or fast potassium activation) that restore the membrane after excitation.
The variable $z$ corresponds to slow ion channels which adapt the excitability of the membrane over longer timescales.  
Finally, the generating can be thought of as a clock that measures progression along the neuronal trajectory.
From a neurophysiological pov, the curve shows how the membrane voltage and ionic currents evolve in time inside a single neuron. This is the basic electrical response unit that when connected in huge networks makes the brain’s computation possible.
If we tune parameters so the model oscillates, then the neuron would fire repetitive spikes, which is the language neurons use to communicate.
\noindent
\begin{remark}
The slow variable can affect the symplectic dynamics of the fast and intermediate ones. On the other hand, \( x \) and \( y \) can influence \( z \) only if the Hamiltonian depends on them in a way such that the Reeb field associated with $z$ picks up that dependence. In this way, there is a one-way geometric coupling: the slow variables shape the fast dynamics directly, while the fast ones can only modulate the slow dynamics indirectly, and never through the symplectic structure \( \Omega \), which is restricted to \( \xi \).
\end{remark}
In the following subsection we proceed to construct a Hamiltonian for the FitzHugh-Nagumo equations \cite{fitz2,fitz3,fitz4}.

\subsubsection*{Construction principle (adjoint/Pontryagin lift in physical time)}

In the theory of cotangent lifts of first--order dynamical systems, a system
\[
\dot{q} = F(q), \qquad q \in \mathbb{R}^n,
\]
admits an adjoint (Pontryagin) lift to the cotangent bundle \( T^\ast\mathbb{R}^n \),
endowed with the canonical symplectic form, generated by the Hamiltonian
\[
H(q,p) = p \cdot F(q)
= \sum_{j=1}^n p_j\,F_j(q).
\]
Hamilton's equations yield
\[
\dot q = \partial_p H = F(q), \qquad
\dot p = -\partial_q H = -(DF(q))^\top p,
\]
so that the original dynamics is recovered in the \emph{configuration} variables, while the \emph{adjoint} equations appear automatically \cite{adjoint,adjoint2}. 

In our case, we consider the extended phase space
\[
M = T^{*}\mathbb{R}^3 \times \mathbb{R}^3,\quad
(x,y,z,p_x,p_y,p_z,t_f,t_i,t_s),
\]
endowed with the canonical $2$--form
\[
\Omega = dx \wedge dp_x + dy \wedge dp_y + dz \wedge dp_z,
\]
together with three closed $1$--forms
\[
\lambda_f = dt_f,\quad \lambda_i = dt_i,\quad \lambda_s = dt_s.
\]
It's clear that  \( (M, \Omega, \lambda_s, \lambda_i, \lambda_f) \) is a 3-cosymplectic manifold, where \( \lambda_s = dt_s, \lambda_i = dt_i, \lambda_f = dt_f \), the Reeb vector fields are \( R_s = \partial_{t_s}, R_i = \partial_{t_i}, R_f = \partial_{t_f} \), and the horizontal and Reeb distributions are given separately by
\begin{align*}
    \xi &= \ker \lambda_f \cap \ker \lambda_i \cap \ker \lambda_s = T\mathbb{R}^3,\\
    \mathcal{R} &= \mathrm{span}\{R_s, R_i, R_f\} = \mathrm{span}\{\partial_{t_s}, \partial_{t_i}, \partial_{t_f}\}.
\end{align*}

Now,
We introduce the notation:
\begin{equation}
 \alpha_f=\frac{\dd t_f}{\dd t}=\frac{1}{\varepsilon},\qquad
        \alpha_i=\frac{\dd t_i}{\dd t}=\frac{1}{\delta},\qquad \alpha_s=\frac{dt_s}{dt}=1.
    \end{equation}
The \emph{horizontal Hamiltonian vector field}
$X_H\in\Gamma(\xi)$ determined by
\begin{equation}\label{eq:qHam}
\ii_{X_H}\Omega \;=\; \dd H - \sum_{\alpha\in\{f,i,s\}} (R_\alpha H)\,\lambda_\alpha,
\qquad
\lambda_\alpha(X_H)=0.
\end{equation}
Given \emph{time-scale coefficients} $\alpha_f,\alpha_i,\alpha_s\in C^\infty(M)$, the
\emph{$q$–evolution field} is
\begin{equation}\label{eq:qEvol}
E_H \;=\; \alpha_f R_f+\alpha_i R_i+\alpha_s R_s + X_H .
\end{equation}
Evolving by $E_H$ in the physical time $t$ encodes the (possibly distinct) clock rates 
$\dot t_f=\alpha_f$, $\dot t_i=\alpha_i$, $\dot t_s=\alpha_s$ (Poincaré time rescaling). For the extended FitzHugh--Nagumo system \cite{adjoint,adjoint2}
\[
\dot x = \frac{1}{\varepsilon}\, f(x,y) = \frac{1}{\varepsilon}\Big(x - \frac{x^3}{3} - y\Big), \qquad
\dot y = \frac{1}{\delta}\, g(x,z) = \frac{1}{\delta}(x + a - z), \qquad
\dot z = h(x,z) = b\, x - c\, z,
\]
the dynamics is generated on \(T^\ast \mathbb{R}^3\) by the \emph{Pontryagin Hamiltonian}

\begin{equation}\label{eq:HP}
\widetilde H(x,y,z,p) =
p_x \Big(\frac{1}{\varepsilon} f(x,y)\Big)
+ p_y \Big(\frac{1}{\delta} g(x,z)\Big)
+ p_z\, h(x,z).
\end{equation}

\subsubsection*{Hamilton's Equations}

In physical time \(t\), Hamilton's equations read
\[
\begin{aligned}
\dot x &= \partial_{p_x} \widetilde H = \frac{1}{\varepsilon} f(x,y),\\
\dot y &= \partial_{p_y} \widetilde H = \frac{1}{\delta} g(x,z),\\
\dot z &= \partial_{p_z} \widetilde H = h(x,z),\\
\dot p_x &= -\partial_x \widetilde H = -\frac{1-x^2}{\varepsilon} p_x - \frac{1}{\delta} p_y - b\, p_z,\\
\dot p_y &= -\partial_y \widetilde H = \frac{1}{\varepsilon} p_x,\\
\dot p_z &= -\partial_z \widetilde H = \frac{1}{\delta} p_y + c\, p_z.
\end{aligned}
\]

The corresponding \emph{clock rates} are encoded by the \(q\)--evolution field
\[
E_{\widetilde H} = \alpha_f R_f + \alpha_i R_i + \alpha_s R_s + X_{\widetilde H}, 
\qquad
\alpha_f = \frac{1}{\varepsilon}, \quad
\alpha_i = \frac{1}{\delta}, \quad
\alpha_s = 1.
\]

Accordingly, \(\dot t_f = 1/\varepsilon\), \(\dot t_i = 1/\delta\), and \(\dot t = 1\), which makes the multi-scale clocks explicit. The Reeb vector fields advance the clocks, while the variables \((x,y,z,p)\) evolve under the horizontal Hamiltonian vector field \(X_{\widetilde H}\).

\subsubsection*{Fast Layer Dynamics (\(\varepsilon \to 0\))}
\noindent
Introduce the fast time \(t_f = t / \varepsilon\), so that \(d/dt_f = \varepsilon\, d/dt\). To leading order, we have
\begin{align*}
x'   &= f(x,y) = x - \frac{x^3}{3} - y, \\
y'   &= \frac{\varepsilon}{\delta} g(x,z) 
     = \frac{\varepsilon}{\delta}(x + a - z) \approx 0, \\
z'   &= \varepsilon h(x,z) = \varepsilon(b\,x - c\,z) \approx 0, \\
p_x' &= -(1-x^2) p_x + O\!\left(\tfrac{\varepsilon}{\delta}, \varepsilon\right), \\
p_y' &= p_x, \\
p_z' &= \tfrac{\varepsilon}{\delta} p_y + \varepsilon c\, p_z \approx 0.
\end{align*}

where the prime (\('\)) denotes \(d/dt_f\).

Consider the critical manifold $\mathcal{C}_{\mathrm{crit}} = \{ f(x,y) = 0 \} = \{ y = \phi(x) := x - \tfrac{x^3}{3} \}$, which is a S-shaped curve in the $(x,y)$ plane. Projecting onto \(\mathcal{C}_{\mathrm{crit}}\), we can get  the reduced pair
\begin{equation}
\dot x = \frac{g(x,z)}{\delta (1-x^2)}, \qquad
\dot z = h(x,z), \qquad
y = \phi(x).
\end{equation}
In fact, \(\dot z = h(x,z)\) and \(y = \phi(x)\) are straightforward. We now explain why 
\(\dot x = \frac{g(x,z)}{\delta (1-x^2)}\).
Since \(y = \phi(x)\) on the critical manifold, the differential relationship along the manifold is
\[
\dot y = \phi'(x) \, \dot x,
\]
where \(\dot y\) denotes the derivative with respect to the original slow time \(t\), and 
\(\phi'(x) = 1 - x^2\).
On the other hand, the original system in slow time gives
\[
\dot y = \frac{1}{\delta} g(x,z).
\]
Combining with the chain rule expression
\[
\dot y = \frac{d}{dt} \phi(x) = \phi'(x) \, \dot x = (1 - x^2) \, \dot x,
\]
we obtain
\[
(1 - x^2) \, \dot x = \frac{1}{\delta} g(x,z) \quad \Rightarrow \quad
\dot x = \frac{g(x,z)}{\delta (1 - x^2)}.
\]

Hence, the reduced Hamiltonian on the reduced cotangent bundle \(T^\ast\{(x,z)\}\) is
\begin{equation}
\widetilde H^{(f)}(x,z; p_x, p_z) = p_x \Big( \frac{g(x,z)}{\delta (1-x^2)} \Big) + p_z\, h(x,z),
\end{equation}
with \(y = \phi(x)\). This can generate the fast-reduced flow  and produce a \(2\)--cosymplectic reduced system up to \(O(\varepsilon)\)  on the reduced cotangent bundle \(T^\ast\{(x,z)\}\).

\subsubsection*{Intermediate layer dynamics ($\delta\to0$)}

We start from the fast-reduced flow of the system
\begin{equation}
\dot x = \frac{g(x,z)}{\delta (1-x^2)}, \qquad
\dot{z} = h(x,z), \qquad
y = \varphi(x) := x - \frac{x^{3}}{3},
\end{equation} and 
we know the intermediate time scale
\[
t_{i} = \frac{t}{\delta}, \qquad R_{i} = \frac{\partial}{\partial t_{i}},
\]
where $R_{i}$ is the Reeb vector field corresponding to the intermediate clock.
At this scale, we have that 
$$z'=\delta h(x,z)\approx 0,\qquad (' = d/dt_{i})$$
thus,
the slow variable $z$ is frozen at leading order (quasi-static approximation).

We can also see that the equation for the fast variable becomes
\begin{equation}
x' = \frac{g(x,z)}{ 1-x^{2}},
\end{equation}
or equivalently in the original time $t$,
\begin{equation}
\dot{x} = \frac{g(x,z)}{\delta (1-x^2)}.
\end{equation}
Since $\delta \to 0$, the right-hand side is large, and $x$ rapidly relaxes to its critical point $x^{*}(z)$ defined by
\begin{equation}
g(x^{*},z) = 0,
\end{equation}
where
\[
g(x,z) = x + a - z,
\]
thus
\[
x^{*}(z) = z - a.
\]
Henceforth, we denote the intermediate layer by the critical manifold
\[
\mathcal{I}_{\mathrm{crit}}
= \bigl\{\, g(x,z) = 0 \,\bigr\}
= \bigl\{\, x = \varphi(z) := z - a \,\bigr\},
\]
whis is a straight line in the $(x,z)$ plane.

On the critical $\mathcal{C}_{\mathrm{crit}}$ we have $y = \phi(x)$, so
\[
y(t) = \phi(x(t)) \quad \Longrightarrow \quad
\dot{y} = \phi'(x)\, \dot{x}
= (1-x^{2})\, \dot{x}.
\]
However, the original system provides
\[
\dot{y} = \frac{1}{\delta}\, g(x,z).
\]
Combining both expressions on the critical manifold $\mathcal{I}_{\mathrm{crit}}\ni x = x^{*}(z)\neq  \pm 1$ yields
\[
(1-(x^{*})^{2})\, \dot{x}
= \frac{1}{\delta}\, g(x^{*},z)
\quad \Longrightarrow \quad
\dot{x} = 0,
\]
because $g(x^{*},z) = 0$. Hence, the fast variable $x$ is equilibrated on intermediate layer $\mathcal{I}_{\mathrm{crit}}$.

Now we can get that
the reduced Hamiltonian in the intermediate layer which depends only on $(z,p_{z})$:
\begin{equation}
\widetilde{H}^{(f_{i})}(z; p_{z}) = p_{z}\,\big((b-c)\,z - a b\big),
\end{equation}
with $y=\phi(\varphi(z))$ and $x=\varphi(z),$
which defines a 1--cosymplectic Hamiltonian system, as there remains only one Reeb vector field (the slow clock $R_{s}$) and one pair of phase space variables $(z,p_{z})$.

The reduced Hamilton's equations give
\begin{align*}
\dot{z} &= \frac{\partial \widetilde{H}^{(f_{i})}}{\partial p_{z}}
= (b-c)\,z - a b,\\
\dot{p_z}&=\frac{\partial \widetilde{H}^{(f_{i})}}{\partial z}=-cp_z,
\end{align*}
which coincides with the result obtained above, with an accuracy of
$
\mathcal{O}(\varepsilon + \delta).
$
The fast and intermediate variables follow as
\[
x(t) = \varphi (z(t)) = z(t) - a,
\qquad
y(t) = \phi(x(t))
= \phi\big(z(t)-a\big)=z(t)-a-\frac{(z(t)-a)^3}{3}.
\]
Thus, the slow flow on the intermediate layer determines the full trajectory through the slaving relations for $(x,y)$.

\subsection{Hamilton-Jacobi for the FitzHugh-Nagumo system}

We now formulate the Hamilton--Jacobi (HJ) problem in the $q$--cosymplectic setting for the extended FitzHugh--Nagumo neuron model. Recall the Pontryagin Hamiltonian
\[
\widetilde H(x,y,z,p) =
\frac{1}{\varepsilon}\, p_x\Big(x - \tfrac{x^3}{3} - y\Big)
+ \frac{1}{\delta}\, p_y\,(x+a-z)
+ p_z\Big(bx - cz\Big),
\]
where $(x,y,z)$ are the configuration variables, $(p_x,p_y,p_z)$ their conjugate momenta, and the small parameters $0<\varepsilon\ll\delta\ll 1$ separate the fast, intermediate and slow dynamics.

Let
\[
M = T^\ast\mathbb{R}^3 \times \mathbb{R}^3, \qquad 
(x,y,z,p_x,p_y,p_z,t_f,t_i,t_s),
\]
endowed with
\[
\Omega = dx \wedge dp_x + dy \wedge dp_y + dz \wedge dp_z, 
\qquad 
\lambda_f = dt_f, \;\; \lambda_i = dt_i, \;\; \lambda_s = dt_s.
\]
The Reeb vector fields are
\[
R_f = \frac{\partial}{\partial t_f}, \quad
R_i = \frac{\partial}{\partial t_i}, \quad
R_s = \frac{\partial}{\partial t_s}.
\]

The HJ problem seeks a section $\gamma:N \to M$, locally of the form
\[
\gamma(x,y,z,t_f,t_i,t_s) = \big(x,y,z,\gamma_x(x,y,z,t),\gamma_y(x,y,z,t),\gamma_z(x,y,z,t),t_f,t_i,t_s\big),
\]
such that
\[
T\gamma \circ X = X_H \circ \gamma.
\]

The isotropy condition $\gamma^\ast \Omega = 0$ implies the existence of a generating function $S$ with
\[
\gamma_x = \partial_x S, \qquad
\gamma_y = \partial_y S, \qquad
\gamma_z = \partial_z S.
\]

Biologically, it is natural to consider the threshold parameter $a$ as constant, while the conductance-related coefficients $b,c$ evolve on the slowest timescale. Thus,
\[
a = \mathrm{const}, \qquad
b = b(t_s), \qquad c = c(t_s).
\]
\noindent
The Pontryagin Hamiltonian becomes
\begin{equation}\label{Hamiltonian}
\widetilde H(x,y,z,p;t_s) =
\frac{1}{\varepsilon}\, p_x\Big(x - \tfrac{x^3}{3} - y\Big)
+ \frac{1}{\delta}\, p_y\,(x+a-z)
+ p_z\Big(b(t_s)\,x - c(t_s)\,z\Big).
\end{equation}

Since $\widetilde H$ now depends explicitly on $t_s$,
\[
R_s(\widetilde H) = \frac{\partial \widetilde H}{\partial t_s} \neq 0,
\]
while $R_f(\widetilde H)=R_i(\widetilde H)=0$. The Hamilton--Jacobi condition
\[
d(\widetilde H \circ \gamma) = \sum_{\alpha\in\{f,i,s\}} (R_\alpha \widetilde H\circ\gamma)\,\gamma^\ast\lambda_\alpha
\]
reduces to
\begin{equation}\label{HJgeom}
d(\widetilde H \circ \gamma) = (R_s \widetilde H\circ\gamma)\, dt_s.
\end{equation}

{Therefore, the Hamilton--Jacobi partial differential equation (HJ PDE) for the generating function \( S(x,y,z,t_s) \) is
\begin{equation}\label{HJnew}
\frac{1}{\varepsilon}\, S_x \Big(x - \tfrac{x^3}{3} - y\Big)
+ \frac{1}{\delta}\, S_y\,(x+a-z)
+ S_z\Big(b(t_s)\,x - c(t_s)\,z\Big)
+ S_{t_s} = 0,
\end{equation}
where the Hamiltonian is given by
\[
\widetilde{H}(x,y,z,p_x,p_y,p_z,t_s) 
= \frac{1}{\varepsilon} p_x \left(x - \frac{x^3}{3} - y\right) 
+ \frac{1}{\delta} p_y (x+a-z) 
+ p_z \big(b(t_s)x - c(t_s)z\big).
\]
We identify the momentum variables with derivatives of the generating function \( S \):
\[
p_x = S_x, \qquad p_y = S_y, \qquad p_z = S_z,
\]
where \( S_x = \frac{\partial S}{\partial x},\; S_y = \frac{\partial S}{\partial y},\; S_z = \frac{\partial S}{\partial z} \). Therefore,
\[
\widetilde{H} \circ \gamma = \widetilde{H}(x,y,z,S_x,S_y,S_z,t_s) 
= \frac{1}{\varepsilon} S_x \left(x - \frac{x^3}{3} - y\right) 
+ \frac{1}{\delta} S_y (x+a-z) 
+ S_z \big(b(t_s)x - c(t_s)z\big).
\]

According to Remark~\ref{REM}, we compute the differentials:
\[
dS_x = S_{xx}\,dx + S_{xy}\,dy + S_{xz}\,dz, \quad
dS_y = S_{yx}\,dx + S_{yy}\,dy + S_{yz}\,dz,
\]
\[
dS_z = S_{zx}\,dx + S_{zy}\,dy + S_{zz}\,dz, \quad
dS_{t_s} = S_{t_s t_s}\,dt_s = 0.
\]

We now compute the differential of each term in the Hamiltonian. First:
\[
d\left(\frac{1}{\varepsilon} S_x \left(x - \frac{x^3}{3} - y\right)\right),
\]
then:
\[
d\left(\frac{1}{\delta} S_y (x+a-z)\right),
\]
and for the third term:
\[
d\left(S_z (b x - c z)\right) = dS_z\,(b x - c z) + S_z\,d(b x - c z).
\]

If we denote by \( d' \) the differential with respect to the spatial variables \( x, y, z \), and \( dt_s \) the differential with respect to \( t_s \), then we can split:
\[
d\left(S_z (b x - c z)\right) = d'S_z\,(b x - c z) + S_z\,(b\,d'x - c\,d'z) + S_z\,(b' x - c' z)\,dt_s,
\]
where \( b' = \frac{db}{dt_s} \), \( c' = \frac{dc}{dt_s} \).

Thus, the full differential becomes:
\begin{footnotesize}
\[
d(\widetilde{H} \circ \gamma) =
d'\left(\frac{1}{\varepsilon} S_x \left(x - \frac{x^3}{3} - y\right)\right)
+ d'\left(\frac{1}{\delta} S_y (x+a-z)\right)
+ d'\left(S_z (b x - c z)\right)
+ S_z(b'x - c'z)\,dt_s.
\]
\end{footnotesize}

Now, the right-hand side of the geometric HJ equation~\eqref{HJgeom} is
\[
R_s \coloneqq \frac{\partial}{\partial t_s}, \qquad
\frac{\partial \widetilde{H}}{\partial t_s} = p_z\big(b'(t_s)x - c'(t_s)z\big),
\]
\[
(R_s \widetilde{H} \circ \gamma)\, dt_s
= S_z\big(b'(t_s)x - c'(t_s)z\big)\, dt_s.
\]

Therefore, equation~\eqref{HJgeom} reduces to
\[
d'\left(
\frac{1}{\varepsilon} S_x \left(x - \frac{x^3}{3} - y\right)
+ \frac{1}{\delta} S_y (x+a-z)
+ S_z (b x - c z)
\right) = 0.
\]

Integrating, we obtain a function of time on the right-hand side. Denoting it by \( F(t_s) \), we get
\begin{equation}\label{check}
\frac{1}{\varepsilon} S_x \left(x - \frac{x^3}{3} - y\right)
+ \frac{1}{\delta} S_y (x+a-z)
+ S_z (b x - c z) = F(t_s).
\end{equation}

This function \( F(t_s) \) could be associated with the time derivative term via
\[
F(t_s) = -\frac{\partial \widehat{S}}{\partial t_s},
\]
thus reproducing equation~\eqref{HJnew}, with \( \widehat{S} = S(x,y,z) - \int F(t_s)\,dt_s \).
}

Equation \eqref{HJnew} is the {biologically compatible Hamilton--Jacobi equation}, where the explicit $\partial_{t_s}S$ term reflects the slow modulation of the neuron by ion concentration or metabolic processes.

The general analytic solution for $S(x,y,z,t_s)$ does not exist in closed form 
because of the cubic nonlinearity and explicit time-dependence in $b,c$. 
The PDE is solved implicitly by the method of characteristics. The method of characteristics converts this PDE into the ODE system:
\[
\begin{cases}
\displaystyle \frac{dx}{ds} = \frac{1}{\varepsilon} \Big(x - \frac{x^3}{3} - y\Big),\\[1em]
\displaystyle \frac{dy}{ds} = \frac{1}{\delta} (x + a - z),\\[0.5em]
\displaystyle \frac{dz}{ds} = b(t_s)\, x - c(t_s)\, z,\\[0.5em]
\displaystyle \frac{dt_s}{ds} = E,\\[0.5em]
\displaystyle \frac{dS}{ds} = 0.
\end{cases}
\]

Here $s$ is the characteristic parameter along which $S$ is constant. The last equation implies that $S$ is constant along each characteristic curve: $S = \text{const}$.
\[
S(x,y,z,t_s) = F(\Gamma_1,\Gamma_2,\Gamma_3),
\]
where $\Gamma_i$ are constants along solutions of the characteristic ODE system.

To check the validity of the Hamilton--Jacobi equation, we design a program in MATLAB.

\subsubsection{Numerical Integration}

To validate the Hamilton--Jacobi equation we would have to
integrate the PDE \eqref{HJnew} for the action $S(x,y,z,t_s)$ directly on a numerical grid,
compute its spatial derivatives $(S_x,S_y,S_z)$, and then compare these gradients
against the adjoint variables $(p_x,p_y,p_z)$.  However, this approach faces severe difficulties. The HJ equation in our setting is a first-order nonlinear PDE posed in
  four variables $(x,y,z,t_s)$. A naive discretization requires a
  four-dimensional grid, which is computationally expensive and memory intensive. Furthermore, the PDE contains the nonlinear term $x - x^3/3 - y$ in the denominator when
  solved for $S_x$, which introduces singularities whenever
  $x - x^3/3 - y \approx 0$. This leads to stiffness in the numerical integration and requires extremely small time steps for stability. Even if $S$ is computed numerically, the quantities of interest are its
  derivatives $(S_x,S_y,S_z)$. Differentiation of noisy numerical solutions is
  highly unstable and amplifies discretization errors, especially in multiple
  spatial dimensions.

Because of these obstacles, a direct PDE integration of $S$ is not practical
for validation purposes. Instead, we exploit the fact that the adjoint
equations and the PDE are equivalent in the sense that the gradient
$(S_x,S_y,S_z)$ evolves under the same dynamics as $(p_x,p_y,p_z)$. By
prescribing initial conditions for $S_y,S_z,S_t$ 







and reconstructing $S_x$
algebraically from the PDE constraint, we obtain a tractable validation
procedure that avoids the instability of full PDE integration.

For this approach, we consider the Hamiltonian \eqref{Hamiltonian} and the Hamilton--Jacobi equation \eqref{HJequationqco}. Hamilton’s equations yield the adjoint (momentum) system
\[
\begin{aligned}
\dot{p}_x &= -\tfrac{1}{\varepsilon}(1-x^2)\,p_x - \tfrac{1}{\delta}p_y - b\,p_z, \\
\dot{p}_y &= \tfrac{1}{\varepsilon} p_x, \\
\dot{p}_z &= \tfrac{1}{\delta} p_y + c\,p_z,
\end{aligned}
\]
coupled with the state equations
\[
\dot{x} = \tfrac{1}{\varepsilon}\Big(x - \tfrac{x^3}{3} - y\Big), 
\quad
\dot{y} = \tfrac{1}{\delta}(x+a-z),
\quad
\dot{z} = bx - cz.
\]

Instead of integrating the ODE for $S_x$, we reconstruct it directly from the Hamilton--Jacobi PDE.  
Given $S_y,S_z,S_{t_s}$ at a point $(x,y,z,t_s)$, we solve algebraically for $S_x$: \
\[
S_x = -\frac{\varepsilon}{\,x - \tfrac{x^3}{3} - y\,}
\left( \frac{1}{\delta} S_y(x+a-z) + S_z(bx-cz) + S_{t_s} \right).
\]

Thus the program only evolves $S_y$ and $S_z$ via
\[
\dot{S}_y = \tfrac{1}{\varepsilon} S_x, 
\qquad
\dot{S}_z = \tfrac{1}{\delta} S_y + c S_z,
\]
while $S_x$ is recomputed at every time step from the PDE constraint, since $S_{t_s}$ is constant and we can choose it zero for simplicity.

The coupled system of ODEs for$(x,y,z,\,p_x,p_y,p_z,\,S_x,S_y,S_z)$
is integrated with the MATLAB solver \texttt{ode45}. This is a variable-step, variable-order Runge--Kutta (4th/5th order Dormand--Prince pair) method.  

We used solver tolerances
\[
\texttt{RelTol} = 10^{-8}, \qquad \texttt{AbsTol} = 10^{-10},
\]
to ensure agreement between the adjoint variables $(p_x,p_y,p_z)$ and the PDE-constructed gradients $(S_x,S_y,S_z)$ to machine precision.  

The validation focuses on comparing $p_x$ (obtained from the adjoint ODE) and $S_x$ (reconstructed from the PDE).  
Since $p_y \equiv S_y$ and $p_z \equiv S_z$ by construction of the ODEs, only the $x$-component provides a non-trivial check.  

Two diagnostics are produced:
\begin{itemize}
  \item A semilogarithmic plot of $|p_x|$ (red solid) and $|S_x|$ (green dashed) versus time, showing exponential growth but perfect overlap.
  \item A ratio plot $p_x/S_x$, which remains identically equal to $1$ up to numerical roundoff errors. The program also computes the maximum deviation
  \[
  \max_t \, |\,p_x(t)/S_x(t) - 1\,|,
  \]
  which was observed to remain at the level of $10^{-11}$--$10^{-12}$ for typical runs.
\end{itemize}

\begin{figure}[h!]
    \centering
    \includegraphics[width=0.7\textwidth]{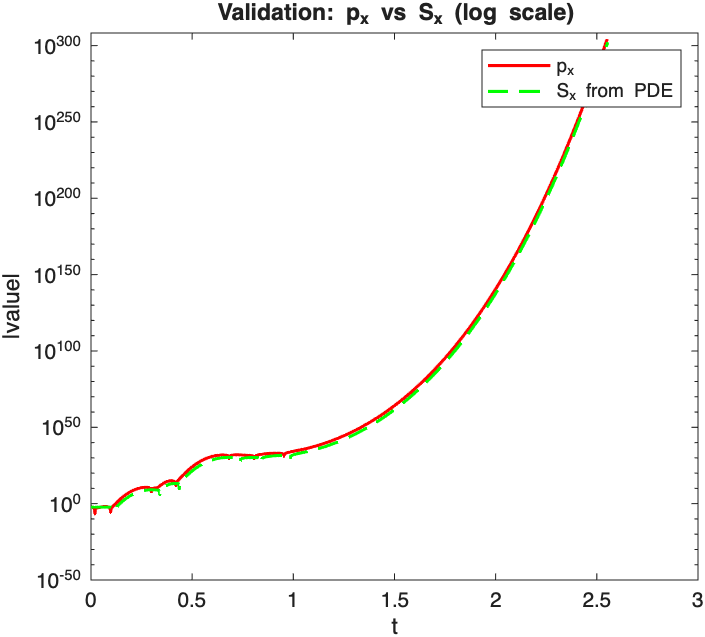}
    \caption{\textbf{Numerical validation of the Hamilton--Jacobi equation.}}
    \label{fig:validation}
\end{figure}

This numerical experiment confirms that the adjoint dynamics and the PDE formulation of the Hamilton--Jacobi equation are consistent: the momentum $p_x$ obtained from Hamilton’s equations coincides with the gradient $S_x$ reconstructed from the PDE constraint. The use of MATLAB’s \texttt{ode45} solver with strict tolerances ensures that deviations are only due to floating-point error.

\subsubsection{Linearization}
There are important situations where analytic solutions of the Hamilton--Jacobi equation can be obtained, as near equilibrium. Let $(x_e(t_s),y_e(t_s),z_e(t_s))$ be a {slowly drifting equilibrium} of the extended FitzHugh--Nagumo system at a fixed slow time $t_s$, i.e.
\begin{align*}
f(x_e,y_e) = x_e - &\frac{x_e^3}{3} - y_e = 0, \quad
g(x_e,z_e) = x_e + a - z_e = 0, \\
&h(x_e,z_e;t_s) = b(t_s)x_e - c(t_s) z_e = 0.
\end{align*}

Introduce shifted coordinates
\[
u =
\begin{pmatrix}
x - x_e(t_s) \\
y - y_e(t_s) \\
z - z_e(t_s)
\end{pmatrix}.
\]

Linearizing the vector field near this equilibrium gives
\[
\dot u = A(t_s) u, \qquad
A(t_s) =
\begin{pmatrix}
\frac{1}{\varepsilon}(1 - x_e(t_s)^2) & -\frac{1}{\varepsilon} & 0 \\
\frac{1}{\delta} & 0 & -\frac{1}{\delta} \\
b(t_s) & 0 & -c(t_s)
\end{pmatrix}.
\]

The linearized Hamilton--Jacobi equation in terms of $S(u,t_s)$ becomes

\[
\frac{1}{\varepsilon} (1 - x_e(t_s)^2) S_u^x \, u_x
- \frac{1}{\varepsilon} S_u^x \, u_y
+ \frac{1}{\delta} S_u^y \, u_x - \frac{1}{\delta} S_u^y \, u_z
+ b(t_s) S_u^z \, u_x - c(t_s) S_u^z \, u_z
+ S_{t_s} = 0,
\]
which is linear in the shifted coordinates $(u_x,u_y,u_z)$ with slowly varying coefficients. Since the coefficients depend smoothly on $t_s$, there exist local analytic solutions of the Hamilton--Jacobi PDE near the slowly drifting equilibrium, which can be expressed as power series in $u$ with coefficients depending on $t_s$.

We look for a generating function $S(u,t_s)$ in the form of a power series in $u$ up to second order:
\[
S(u,t_s) = \frac{1}{2} u^\top P(t_s) u + Q(t_s)^\top u + R(t_s),
\]
where $P(t_s) \in \mathbb{R}^{3\times 3}$ is symmetric, $Q(t_s) \in \mathbb{R}^3$, and $R(t_s) \in \mathbb{R}$.  

The derivatives are
\[
S_u = P(t_s) u + Q(t_s), \qquad
S_{t_s} = \frac{1}{2} u^\top \dot P(t_s) u + \dot Q(t_s)^\top u + \dot R(t_s),
\]
where $\dot P = dP/dt_s$, etc.  Substituting into the linearized HJ PDE yields equations for coefficients:

\begin{itemize}
    \item \textbf{Quadratic terms:} $u^\top$-dependent terms give a matrix ODE for $P(t_s)$:
    \[
    \dot P(t_s) + A(t_s)^\top P(t_s) + P(t_s) A(t_s) = 0.
    \]
    \item \textbf{Linear terms:} $u$-dependent terms give a vector ODE for $Q(t_s)$:
    \[
    \dot Q(t_s) + A(t_s)^\top Q(t_s) = 0.
    \]
    \item \textbf{Constant term:} $R(t_s)$ satisfies
    \[
    \dot R(t_s) = 0 \quad \text{(or determined by boundary condition).}
    \]
\end{itemize}

Hence, the local analytic solution of the nonautonomous Hamilton--Jacobi PDE near the slowly drifting equilibrium is
\[
S(x,y,z,t_s) \approx \frac{1}{2} 
\begin{pmatrix} x-x_e \\ y-y_e \\ z-z_e \end{pmatrix}^\top
P(t_s)
\begin{pmatrix} x-x_e \\ y-y_e \\ z-z_e \end{pmatrix}
+ Q(t_s)^\top 
\begin{pmatrix} x-x_e \\ y-y_e \\ z-z_e \end{pmatrix} + R(t_s),
\]
with $P(t_s)$, $Q(t_s)$, $R(t_s)$ evolving according to the above ODEs. This solution is analytic in $(x,y,z)$ near the equilibrium and encodes the slow-time modulation through $t_s$.

\section*{Acknowledgements}
Cristina Sardón acknowledges Programa Propio of Universidad Politécnica de Madrid
for the granting of financial support for research purposes at UCSD the summer of 2025
and gratefully acknowledges Professor Melvin Leok for hosting him at UC San Diego
and for his support during this work. Xuefeng Zhao gratefully acknowledges the support
from the National Natural Science Foundation of China (Grant No. 12401234).

\end{document}